\documentclass{article}

\usepackage[totalwidth=450pt,totalheight=640pt]{geometry}

\usepackage{graphicx}
\usepackage{color}
\usepackage{amsmath}
\usepackage{amssymb}
\usepackage{amsfonts}
\usepackage{amsthm}
\usepackage{microtype}

\newtheorem{theorem}{Theorem}
\newtheorem{proposition}{Proposition}
\newtheorem{lemma}{Lemma}
\newtheorem{fact}{Fact}

\DeclareRobustCommand{\DE}[3]{#2}

\begin{document}

\title{SDDs are Exponentially More Succinct than OBDDs}

\author{Simone Bova\\Technische Universit{\"a}t Wien\\ \texttt{simone.bova@tuwien.ac.at}}

\date{}

\maketitle

\begin{abstract}
Introduced by Darwiche \cite{Darwiche11}, sentential decision diagrams (SDDs) 
are essentially as tractable as ordered binary decision diagrams (OBDDs), 
but tend to be more succinct \emph{in practice}.  This makes SDDs a prominent representation language, 
with many applications in artificial intelligence and knowledge compilation.

We prove that SDDs are more succinct than OBDDs  
also \emph{in theory}, 
by constructing a family of boolean functions 
where each member has polynomial SDD size 
but exponential OBDD size.  This 
exponential separation 
improves a quasipolynomial separation recently established by Razgon \cite{Razgon13}, 
and settles an open problem in knowledge compilation \cite{Darwiche11}.
\end{abstract}

\section{Introduction}

The idea of \emph{knowledge compilation}  
is to deal with the intractability 
of certain computational tasks 
on a knowledge base by compiling it 
into a different data structure  
where the tasks are feasible.  
The choice of the target data structure  
involves an unavoidable trade-off between succinctness and tractability.

Darwiche and Marquis \cite{DarwicheM02} systematically investigated 
this trade-off in the fundamental case where the knowledge bases are boolean functions 
and the data structures are classes of boolean circuits (\emph{representation languages}).  

In their setting, \emph{decomposable negation normal forms (DNNFs)} 
and \emph{ordered binary decision diagrams (OBDDs)} 
arise as benchmark  
languages for succinctness and tractability respectively \cite{Darwiche01,DarwicheM02}. 
On the one hand, DNNFs are exponentially more succinct than OBDDs; 
moreover, in contrast to OBDDs, they implement efficiently conjunctive normal forms 
of small treewidth \cite{Darwiche01,RazgonP13,OztokD14,Razgon14}.  On the other hand, 
the vast applicability of OBDDs in verification and synthesis 
relies on the tractability of equivalence testing (speeded up by canonicity) 
and boolean combinations, which DNNFs lack \cite{DarwicheM02}.

This gap between DNNFs (succinct but hard) and OBDDs (verbose but tractable) 
led to the quest for intermediate languages 
exponentially more succinct than, 
but essentially as tractable as, OBDDs.  

Introduced by Darwiche \cite{Darwiche11}, 
\emph{sentential decision diagrams (SDDs)} 
are a most prominent candidate 
to narrow the gap between DNNFs and OBDDs.  They are designed 
by strengthening the decomposability property \cite{PipatsrisawatD08}
and further imposing a very strong form of determinism \cite{PipatsrisawatD10}.  The resulting language 
can implement decisions of the form
\begin{equation}\label{eq:sent-dec-intro}
\bigvee_{i=1}^m P_i(X) \wedge S_i(Y)\text{,} 
\end{equation}
where $X$ and $Y$ are disjoint sets of variables nicely structured by an underlying 
\emph{variable tree}, 
and the subcircuits $P_1,\ldots,P_m$, called \emph{primes},\footnote{The $S_i$'s are called \emph{subs}.} implement an exhaustive case distinction 
into exclusive and consistent cases.\footnote{Formally, 
the models of $P_1,\ldots,P_m$ partition the set of assignments of $X$ to $\{0,1\}$ into $m$ nonempty blocks; see Section~\ref{sect:back}.} 
Binary (or Shannon) decisions in OBDDs boil down to very special sentential decisions having the form
$$(\neg x \wedge S_1(Y)) \vee (x \wedge S_2(Y))\text{,}$$
where the variable $x$ is not in the variable set $Y$. 

Indeed, SDDs properly contain OBDDs, 
and hence are at least as succinct as OBDDs, 
while preserving tractability of all key tasks that are tractable on OBDDs.  
For this reason, they have been used in a variety 
of applications in artificial intelligence and probabilistic reasoning, 
as reported, for instance, by \cite{BroeckDarwiche15,OztokDarwiche15}.

Not only SDDs are as tractable as OBDDs, 
but they also tend to be more succinct than OBDDs in practice; 
in fact, knowledge compilers often produce much smaller SDDs than OBDDs by heuristically 
leveraging the additional flexibility of variable trees in SDDs 
with respect to variable orderings in OBDDs \cite{ChoiD13,OztokDarwiche15}. 

Nonetheless,  
the basic theoretical question about the relative succinctness of OBDDs and SDDs 
has been open since Darwiche introduced SDDs \cite{Darwiche11,Razgon13}:
\begin{quote}
\textit{Are SDDs exponentially more succinct than OBDDs?} 
\end{quote}
The results in the literature did not even exclude the possibility for OBDDs to polynomially simulate SDDs \cite{XueCD12}, 
until recently Razgon proved a quasipolynomial separation \cite{Razgon13}.  The above question stands, 
though, as for instance OBDDs could still quasipolynomially simulate SDDs.

\paragraph{Contribution.}   We prove in this article that 
SDDs are exponentially more succinct than OBDDs.  Thus, in particular, 
OBDDs cannot quasipolynomially simulate SDDs.  

More precisely, 
\emph{we construct an infinite family of boolean functions such that 
every member of the family 
has polynomial compressed SDD size but exponential OBDD size} (Theorem~\ref{th:exp-sep}).

\emph{Compressed} SDDs contain OBDDs,\footnote{More precisely, compressed SDDs contain reduced OBDDs; 
see \cite[Definition~1.3.2]{Wegener00}.} and are regarded as a natural SDD class because of their \emph{canonicity}: two compressed SDDs computing the same function 
are syntactically equal up to syntactic manipulations preserving polynomial size \cite{Darwiche11}.  
The restriction to compressed SDDs makes our result stronger, 
because general SDDs are believed (despite not known) to be exponentially more succinct than compressed SDDs \cite{BroeckDarwiche15}.

We separate compressed SDDs and OBDDs by a function, 
which we call the \emph{generalized hidden weighted bit} function because, 
indeed, it contains the \emph{hidden weighted bit} function (HWB) as a subfunction.  
HWB is perhaps the simplest function known to be hard on OBDDs \cite{Bryant86}:
it computes the subsets of $\{1,\ldots,n\}$ 
having size $i$ and containing the number $i$, 
for $i=1,\ldots,n$.

It turns out that HWB itself has small (uncompressed) SDDs (Theorem~\ref{th:sdd-size}), 
which immediately separates SDDs and OBDDs.  The construction, 
a slight variation of which gives the compressed case (Lemma~\ref{lemma:sdd-size} and Lemma~\ref{lemma:obdd-size}), 
is based on the following two observations.  

The first observation is that HWB 
can be expressed 
as a sentential decision of the form (\ref{eq:sent-dec-intro}) by distinguishing the following primes:
\begin{itemize}
\item for $i=1,\ldots,n$, the subsets of size $i$ containing the number $i$ 
(each of these $n$ primes is taken by HWB, so their subs will be equivalent to $\top$);
\item the empty subset, 
and the subsets of size $i$ not containing the number $i$ for $i=1,\ldots,n-1$ (none of these $n$ primes is taken by HWB, 
so their subs will be equivalent to $\bot$).
\end{itemize}
The second observation is that each of the above primes has small OBDD size under any variable ordering (Proposition~\ref{prop:obdd-pi}).  With these two observations 
it is fairly straightforward to implement the hidden weighted bit function 
by a small (uncompressed) SDD (Theorem~\ref{th:sdd-size}).  

A direct inspection of our construction allows to straightforwardly derive 
some facts about compression previously observed in the literature \cite{BroeckDarwiche15}, 
namely that the SDD size may increase exponentially either by compressing SDDs over fixed variable trees, 
or by conditioning (unboundedly many variables) over fixed variable trees (see Section~\ref{sect:disc}).

\paragraph{Organization.} The article is organized as follows.  In Section~\ref{sect:back} 
we present the technical background, 
culminating in the quasipolynomial separation of SDDs and OBDDs proved by Razgon 
(Theorem~\ref{th:raz}). In Section~\ref{sect:main}, 
we separate (uncompressed) SDDs and OBDDs by the hidden weighted bit function (Theorem~\ref{th:sdd-size}) 
and then modify the construction to separate compressed SDDs and OBDDs (Theorem~\ref{th:exp-sep}).  
We discuss our results in Section~\ref{sect:disc}. 

\section{Background}\label{sect:back}

We collect background notions and facts from the literature \cite{DarwicheM02,PipatsrisawatD08,Darwiche11,Razgon13}.

\paragraph{Structured Deterministic NNFs.} 
Let $X$ be a finite set of variables.  
Let $C$ be a boolean circuit on input variables $X$, 
built using fanin $0$ constant gates (labelled by $\bot$ or $\top$), 
fanin $1$ negation gates (labelled by $\neg$), 
and unbounded fanin disjunction and conjunction gates (labelled by $\vee$ and $\wedge$).  
The unique sink node (outdegree $0$) 
in the underlying directed acyclic graph (DAG) of $C$ 
is called the output gate of $C$; 
source nodes (indegree $0$) are called input gates, 
and are labelled by constants or variables in $X$; in particular, 
$C$ is allowed to not read some of the variables in $X$, see Figure~\ref{fig:struct} (left).

A boolean circuit $C$ on variables $X$ is in \emph{negation normal form}, in short an \emph{NNF}, 
if the gates labelled by $\neg$ have wires only from input gates.  Without loss of generality 
we assume that NNFs have input gates labelled by constants or literals on variables in $X$ 
(and no internal gates labelled by $\neg$).

As usual, an NNF $C$ on input variables $X$ computes a boolean function $f \colon \{0,1\}^X \to \{0,1\}$; 
in this case we also write $C \equiv f$.  Two NNFs $C$ and $C'$ on the same input variables are equivalent 
if they compute the same boolean function; again we write $C \equiv C'$.  

The \emph{size} of an NNF $C$, in symbols $\mathrm{size}(C)$, 
is the number of arcs in its underlying DAG.  
Let $f$ be a boolean function 
and let $\mathcal{L}$ be a class of NNFs.  
The \emph{size of $f$ relative to $\mathcal{L}$} 
(or, in short, the \emph{$\mathcal{L}$ size of $f$}), 
denoted by $\mathcal{L}(f)$, 
is equal to the minimum over the sizes of all circuits in $\mathcal{L}$ computing $f$:
$$\mathcal{L}(f)=\mathrm{min}\{ \mathrm{size}(C) \colon C \in \mathcal{L}, C \equiv f\}\text{.}$$

Let $C$ be an NNF on input variables $X$, 
and let $g$ be a gate of $C$.  
We denote by $C_{g}$ the subcircuit of $C$ having $g$ as its output gate, 
that is, the circuit whose underlying DAG is 
the subgraph of the underlying DAG of $C$ 
induced by the nodes having a directed path to $g$ (labelled as in $C$).

\smallskip 

An NNF $C$ on input variables $X$ is \emph{deterministic} 
if, for every $\vee$-gate $g$ in $C$, say of the form $\bigvee_{i=1}^m g_i$, 
it holds that $$C_{g_i} \wedge C_{g_j} \equiv \bot$$ for all $1 \leq i < j \leq m$, 
where we formally regard $C_{g_i}$, $C_{g_j}$, and $\bot$ as NNFs on input variables $X$.
We denote by $\mathcal{NNF}_d$ the class of all deterministic NNFs.

\smallskip 

Let $Y$ be a finite nonempty set of variables.  
A \emph{variable tree} (in short, a \emph{vtree}) for the variable set $Y$ is a 
rooted, full, ordered, binary tree $T$ 
whose leaves correspond bijectively to $Y$; indeed, 
we identify each leaf in $T$ with the variable in $Y$ it corresponds to.  

Let $v$ be an internal node of the vtree $T$.  We let
$v_l$ and $v_r$ denote respectively the left and right child of $v$, 
and $T_{v}$ denote the subtree of $T$ rooted at $v$.  
We also let $Y_v \subseteq Y$ denote (the variables corresponding to) the leaves 
of $T_v$; clearly $T_v$ is a vtree for the variable set $Y_v$.

\smallskip   

Let $C$ be an NNF on input variables $X$, 
and let $T$ be a vtree for the variable set $Y$.  

We say that $C$ \emph{respects $T$} if the following holds.  
First, every $\wedge$-gate $g$ in $C$ has fanin exactly $2$.  
Second, let $g$ be an $\wedge$-gate in $C$ having wires from gates $h_1$ and $h_2$. 
Then there exists an internal node $v$ in $T$ such that 
the input gates of the subcircuit $C_{h_1}$ 
mention only variables in $T_{v_l}$ 
and the input gates of the subcircuit $C_{h_2}$ 
mention only variables in $T_{v_r}$.  In this case, we also say that $g$ respects $v$.

Note that, in particular, 
the sets of variables mentioned by $C_{h_1}$ 
and $C_{h_2}$ are disjoint; it follows that $C$ is decomposable \cite{Darwiche01}.  
Also note that, by definition, if an NNF reading all the variables in a set $X$
is structured by a vtree for the variable set $Y$, 
then $X \subseteq Y$ and the inclusion can be strict; 
see Figure~\ref{fig:struct}.  This feature is crucial in our construction 
(see, for instance, the proof of Theorem~\ref{th:sdd-size}).

A \emph{structured NNF} is an NNF respecting some vtree.  See Figure~\ref{fig:struct}.  
We denote by $\mathcal{NNF}_s$ the class of all structured NNFs.  

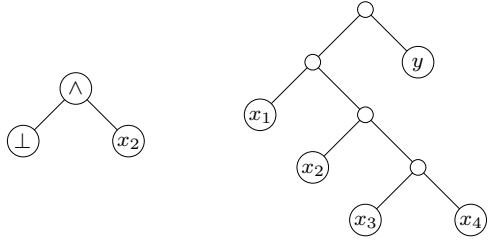
\begin{figure}[t!]
\centering
\begin{picture}(0,0)%
\includegraphics{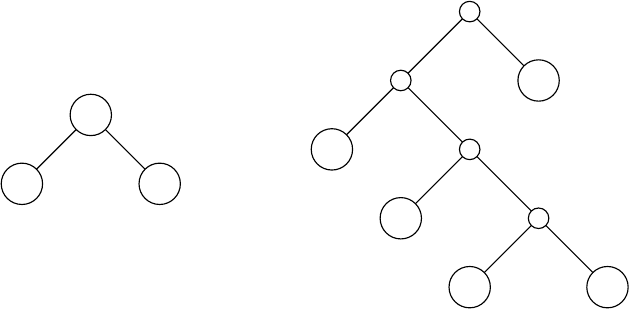}%
\end{picture}%
\setlength{\unitlength}{2901sp}%
\begingroup\makeatletter\ifx\SetFigFont\undefined%
\gdef\SetFigFont#1#2#3#4#5{%
  \reset@font\fontsize{#1}{#2pt}%
  \fontfamily{#3}\fontseries{#4}\fontshape{#5}%
  \selectfont}%
\fi\endgroup%
\begin{picture}(4111,2018)(8858,-429)
\put(9805,344){\makebox(0,0)[lb]{\smash{{\SetFigFont{8}{9.6}{\rmdefault}{\mddefault}{\updefault}{\color[rgb]{0,0,0}$x_2$}%
}}}}
\put(12323,1019){\makebox(0,0)[lb]{\smash{{\SetFigFont{8}{9.6}{\rmdefault}{\mddefault}{\updefault}{\color[rgb]{0,0,0}$y$}%
}}}}
\put(11830,-331){\makebox(0,0)[lb]{\smash{{\SetFigFont{8}{9.6}{\rmdefault}{\mddefault}{\updefault}{\color[rgb]{0,0,0}$x_3$}%
}}}}
\put(9384,793){\makebox(0,0)[lb]{\smash{{\SetFigFont{8}{9.6}{\rmdefault}{\mddefault}{\updefault}{\color[rgb]{0,0,0}$\wedge$}%
}}}}
\put(8922,333){\makebox(0,0)[lb]{\smash{{\SetFigFont{8}{9.6}{\rmdefault}{\mddefault}{\updefault}{\color[rgb]{0,0,0}$\bot$}%
}}}}
\put(10930,569){\makebox(0,0)[lb]{\smash{{\SetFigFont{8}{9.6}{\rmdefault}{\mddefault}{\updefault}{\color[rgb]{0,0,0}$x_1$}%
}}}}
\put(11380,119){\makebox(0,0)[lb]{\smash{{\SetFigFont{8}{9.6}{\rmdefault}{\mddefault}{\updefault}{\color[rgb]{0,0,0}$x_2$}%
}}}}
\put(12730,-331){\makebox(0,0)[lb]{\smash{{\SetFigFont{8}{9.6}{\rmdefault}{\mddefault}{\updefault}{\color[rgb]{0,0,0}$x_4$}%
}}}}
\end{picture}%
\caption{A circuit on input variables $\{x_2,x_4\}$ on the left (in the underlying DAG, 
the edges are oriented upwards), respecting the vtree for the variable set $\{x_1,x_2,x_3,x_4,y\}$ on the right.  
The left subtree is a vtree for the variable set $\{x_1,x_2,x_3,x_4\}$, 
and the right subtree is a vtree for the variable set $\{y\}$.  
The $\wedge$-gate in the circuit respects the root of the vtree.}
\label{fig:struct}
\end{figure}

\paragraph{SDDs and OBDDs.}  
A \emph{sentential decision diagram (SDD) $C$ respecting a vtree $T$}   
is defined inductively as follows.
\begin{itemize}
\item $C$ is a single gate labelled by a literal on a variable $x$, 
and $x$ is in the variable set of $T$.
\item $C$ is a single gate labelled by a constant, and $T$ is any vtree.
\item $C$ is formed by an output gate $g$ labelled by $\vee$, 
with $m \geq 2$ wires from gates $g_1,\ldots,g_m$ labelled by $\wedge$, 
where each $g_i$ has wires from two gates $p_i$ and $s_i$, that is,
\begin{equation}\label{eq:or-node}
C=\bigvee_{i=1}^m C_{p_i} \wedge C_{s_i}\text{,} 
\end{equation}
such that for some internal node $v$ of $T$ the following holds ($i=1,\ldots,m$):
\begin{description}
\item[{\bf(S1)}] $C_{p_i}$ is an SDD respecting a subtree of $T_{v_l}$.
\item[{\bf(S2)}]  $C_{s_i}$ is an SDD respecting a subtree of $T_{v_r}$.
\item[{\bf(S3)}] $C_{p_i} \not\equiv \bot$.
\item[{\bf(S4)}] $C_{p_i} \wedge C_{p_j} \equiv \bot$ ($1 \leq i<j \leq m$).
\item[{\bf(S5)}] $\bigvee_{i=1}^m C_{p_i} \equiv \top$. 
\end{description}
\end{itemize}
In the equivalences in (S3)-(S5), we formally regard the $C_{p_i}$'s, $\bot$ and $\top$ 
as NNFs on variables $Y_{v_l}$.  In words, conditions (S3)-(S5) say that 
the $C_{p_i}$'s define a partition of $\{0,1\}^{Y_{v_l}}$ 
into $m$ nonempty blocks, where the $i$th block contains exactly 
the models of $C_{p_i}$ ($i=1,\ldots,m$).

An SDD is an SDD respecting some vtree.  We let $\mathcal{SDD}$ 
denote the class of all SDDs.

An SDD $C$ is called \emph{compressed} if the following holds.  
Let $h$ be an $\vee$-gate of $C$, 
so that $h=\bigvee_{i=1}^{m'} C_{p'_i} \wedge C_{s'_i}$ specified as in 
(\ref{eq:or-node}) relative to some node $v'$ in $T$. Then
\begin{description}
\item[{\bf(C)}] $C_{s'_i} \not\equiv C_{s'_j}$ ($1 \leq i<j \leq m'$),
\end{description}
where we formally regard $C_{s'_i}$ as an NNF on variables $Y_{v'_r}$ for $i=1,\ldots,m'$.  
We let $\mathcal{SDD}_c$ denote the class of all compressed SDDs.

\smallskip
An \emph{ordered binary decision diagram (OBDD)} is a compressed SDD 
respecting a \emph{right-linear} vtree $T$ (that is, 
where each left child is a leaf); see Figure~\ref{fig:right-linear-tree}.  
We let $\mathcal{OBDD}$ denote the class of all OBDDs.\footnote{Reduced OBDDs as usually defined in the literature \cite[Definition~1.3.2]{Wegener00} 
are indeed compressed SDD respecting right-linear vtrees \cite[Section~6]{Darwiche11}.}

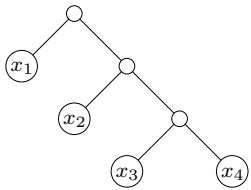
\begin{figure}[t!]
\centering
\begin{picture}(0,0)%
\includegraphics{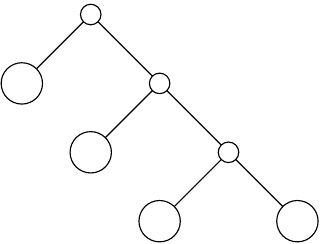}%
\end{picture}%
\setlength{\unitlength}{2901sp}%
\begingroup\makeatletter\ifx\SetFigFont\undefined%
\gdef\SetFigFont#1#2#3#4#5{%
  \reset@font\fontsize{#1}{#2pt}%
  \fontfamily{#3}\fontseries{#4}\fontshape{#5}%
  \selectfont}%
\fi\endgroup%
\begin{picture}(2086,1568)(10883,-429)
\put(11830,-331){\makebox(0,0)[lb]{\smash{{\SetFigFont{8}{9.6}{\rmdefault}{\mddefault}{\updefault}{\color[rgb]{0,0,0}$x_3$}%
}}}}
\put(10930,569){\makebox(0,0)[lb]{\smash{{\SetFigFont{8}{9.6}{\rmdefault}{\mddefault}{\updefault}{\color[rgb]{0,0,0}$x_1$}%
}}}}
\put(11380,119){\makebox(0,0)[lb]{\smash{{\SetFigFont{8}{9.6}{\rmdefault}{\mddefault}{\updefault}{\color[rgb]{0,0,0}$x_2$}%
}}}}
\put(12730,-331){\makebox(0,0)[lb]{\smash{{\SetFigFont{8}{9.6}{\rmdefault}{\mddefault}{\updefault}{\color[rgb]{0,0,0}$x_4$}%
}}}}
\end{picture}%

\caption{A right-linear vtree; its left first traversal induces the variable ordering $x_1<x_2<x_3<x_4$.}
\label{fig:right-linear-tree}
\end{figure}

Let $C$ be an OBDD respecting a vtree $T$, 
and let $\sigma=x_1<\cdots<x_n$ be the variable ordering 
induced by a left first traversal of $T$; in this case, 
we also say that $C$ respects $\sigma$.  
For an ordering $\sigma$ of a set of variables, 
we let $\mathcal{OBDD}_\sigma$ denote the class of all OBDDs respecting $\sigma$.

\paragraph{Quasipolynomial Separation.} It follows from the definitions that 
\begin{equation}\label{eq:inclusions}
\mathcal{OBDD} \subseteq \mathcal{SDD}_c \subseteq \mathcal{SDD} \subseteq \mathcal{NNF}_{s} \cap \mathcal{NNF}_{d}\text{} 
\end{equation}
which raises the natural question how OBDDs and SDDs 
are related in succinctness; indeed, the quest for the relative succinctness of 
OBDDs and SDDs has been an open problem in knowledge compilation since 
Darwiche introduced SDDs \cite{Darwiche11}.  

Recently, Razgon \cite[Corollary~3]{Razgon13} has  
established a \emph{quasipolynomial separation} of OBDDs from compressed SDDs.

\begin{theorem}[Razgon]\label{th:raz}
There exists an unbounded arity class of boolean functions $\mathcal{F}$ 
such that every arity $n$ function $f \in \mathcal{F}$ has $\mathcal{SDD}_c$ size in $O(n^3)$ 
and $\mathcal{OBDD}$ size in $n^{\Omega(\log n)}$.
\end{theorem}

We remark that the restriction to compressed SDDs in the above statement is nontrivial; 
to the best of our knowledge, compressed SDDs 
might be exponentially more succinct than uncompressed SDDs \cite{BroeckDarwiche15}; 
see also the discussion in Section~\ref{sect:disc}.

\section{Exponential Separation}\label{sect:main}

The quasipolynomial separation stated in Theorem~\ref{th:raz}  
implies that OBDDs do not simulate SDDs in polynomial size, 
but leaves open the possibility for OBDDs to simulate SDDs in quasipolynomial size.  
In this section  
we exclude this possibility by establishing 
an \emph{exponential separation}
of OBDDs from compressed SDDs.

\paragraph{Hidden Weighted Bit.} The separation is obtained by (a variant of) the \emph{hidden weighted bit} function 
$$\mathrm{HWB}_n(x_1,\ldots,x_n)\text{,}$$ 
that is the boolean function on $n$ inputs $x_1,\ldots,x_n$ such that,  
for all assignments $f \colon \{x_1,\ldots,x_n\} \to \{0,1\}$, 
it holds that $f$ is a model of $\mathrm{HWB}_n$ if and only if $f(x_1)+\cdots+f(x_n)=i$ 
and $f(x_i)=1$ ($i \geq 1$).  

It is well known that the hidden weighted bit function has exponential OBDD size \cite{Bryant86}.  

\begin{theorem}[Bryant]\label{th:bryant}
The $\mathcal{OBDD}$ size of $\mathrm{HWB}_n$ is $2^{\Omega(n)}$. 
\end{theorem}

Intuitively, a model of $\mathrm{HWB}_n$ is a subsets of $\{1,\ldots,n\}$ 
of size $i$ containing the number $i$, for $i=1,\ldots,n$.  For instance, 
$\mathrm{HWB}_2(1,0)=1$, because the set $\{1\}$ has size $1$ and contains the number $1$, 
and $\mathrm{HWB}_2(0,1)=0$, because the set $\{2\}$ has size $1$ but does not contain the number $1$.

The simple but crucial observation underlying our construction is 
that the models of $\mathrm{HWB}_n$ can be decided arguing by cases, as follows:  
If $S$ is a subset of $\{1,\ldots,n\}$ of size $i$, 
then $S$ is a model of $\mathrm{HWB}_n$ if and only if $i \in S$ ($i=1,\ldots,n$).  
With this insight it is not hard to setup an exhaustive and exclusive 
case distinction equivalent to $\mathrm{HWB}_n$; 
the key observation is that each individual case in the distinction 
is computable by a small OBDD with respect to any variable ordering.

We formalize the above intuition.  For $i \in \{0,1,\ldots,n\}$, 
let $$E^i_n(x_1,\ldots,x_n)$$ be the boolean function 
on $n$ inputs $x_1,\ldots,x_n$ such that,  
for all assignments $f \colon \{x_1,\ldots,x_n\} \to \{0,1\}$, 
it holds that $f$ is a model of $E^i_n$ if and only if $f(x_1)+\cdots+f(x_n)=i$. Hence 
$E^i_n$ computes the subsets of $\{1,\ldots,n\}$ 
of size $i$ ($i \geq 0$).  Let now 
\begin{equation}\label{eq:family-P}
\mathcal{P}_n=\{P_0,P_n\} \cup \{P_{i,0},P_{i,1} \colon i=1,\ldots,n-1\} 
\end{equation}
be the family of $2n$ boolean functions, each over the variables $\{x_1,\ldots,x_n\}$, 
defined as follows:
\begin{itemize}
\item $P_{0} \equiv E^0_n$
\item $P_{n} \equiv E^n_n$
\end{itemize} 
and for $i=1,\ldots,n-1$ let 
\begin{itemize}
\item $P_{i,0} \equiv E^i_n \wedge \neg x_i$ 
\item $P_{i,1} \equiv E^i_n \wedge x_i$ 
\end{itemize}
See Figure~\ref{fig:ein-proj} for an illustration. 

Each function in $\mathcal{P}_n$ computes a family of subsets of $\{1,\ldots,n\}$.  
Namely, $P_0$ computes the empty subset, 
$P_n$ computes $\{1,\ldots,n\}$, 
$P_{i,0}$ computes the subsets of $\{1,\ldots,n\}$ of size $i$ not containing the number $i$, 
and $P_{i,1}$ computes the subsets of $\{1,\ldots,n\}$ of size $i$ containing the number $i$ ($i=1,\ldots,n-1$).  

It is readily observed that the members of 
$\mathcal{P}_n$ partition the powerset of $\{1,\ldots,n\}$ in nonempty blocks. 
Formally,

\begin{fact}\label{fact:part}
Let $\mathcal{P}_n$ be as in (\ref{eq:family-P}), 
and let $P,P' \in \mathcal{P}_n$ with $P \neq P'$.
\begin{itemize}
\item $P \not\equiv \bot$.
\item $P \wedge P' \equiv \bot$.
\item $\bigvee_{P \in \mathcal{P}_n}P \equiv \top$.
\end{itemize}
\end{fact}   

We now establish the key property, 
that each member of $\mathcal{P}_n$ is computable by a small OBDD 
with respect to any variable ordering.  

First consider the functions $E^i_n$.  An OBDD computing $E^i_n$ with respect to the variable ordering $\sigma=x_1<\cdots<x_n$ 
is displayed in Figure~\ref{fig:ein-obdd} for the case $n=4$ and $i=2$.  Generalizing the construction, 
we have that an OBDD $C$ computing $E^i_n$ and respecting $\sigma$ 
has at most $1+2+\cdots+n=n(n+1)/2$ decision nodes, 
each contributing $6$ wires in the circuit; hence $C$ has size $O(n^2)$.  

Since $E^i_n$ is symmetric \cite[Definition~2.3.2 and Lemma~4.7.1]{Wegener00}, the following holds.

\begin{proposition}\label{prop:obdd-ein}
Let $\sigma$ be an ordering of $x_1,\ldots,x_n$.  
The $\mathcal{OBDD}_\sigma$ size of $E^i_n$ is $O(n^2)$. 
\end{proposition}

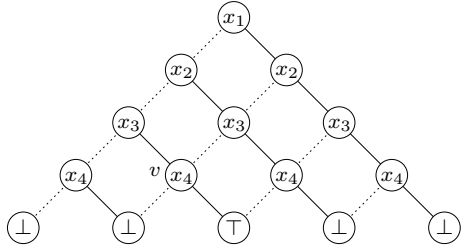
\begin{figure}[t!]
\centering
\begin{picture}(0,0)%
\includegraphics{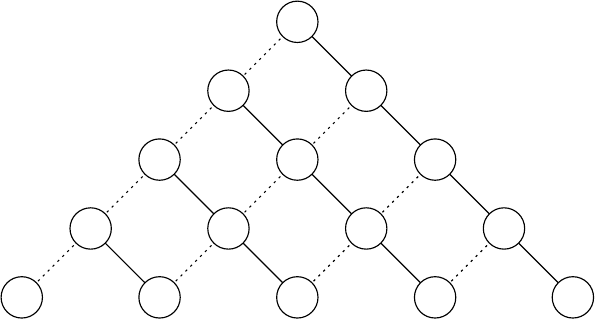}%
\end{picture}%
\setlength{\unitlength}{2901sp}%
\begingroup\makeatletter\ifx\SetFigFont\undefined%
\gdef\SetFigFont#1#2#3#4#5{%
  \reset@font\fontsize{#1}{#2pt}%
  \fontfamily{#3}\fontseries{#4}\fontshape{#5}%
  \selectfont}%
\fi\endgroup%
\begin{picture}(3886,2086)(9533,-879)
\put(11380,1019){\makebox(0,0)[lb]{\smash{{\SetFigFont{8}{9.6}{\rmdefault}{\mddefault}{\updefault}{\color[rgb]{0,0,0}$x_1$}%
}}}}
\put(10930,569){\makebox(0,0)[lb]{\smash{{\SetFigFont{8}{9.6}{\rmdefault}{\mddefault}{\updefault}{\color[rgb]{0,0,0}$x_2$}%
}}}}
\put(11830,569){\makebox(0,0)[lb]{\smash{{\SetFigFont{8}{9.6}{\rmdefault}{\mddefault}{\updefault}{\color[rgb]{0,0,0}$x_2$}%
}}}}
\put(11380,119){\makebox(0,0)[lb]{\smash{{\SetFigFont{8}{9.6}{\rmdefault}{\mddefault}{\updefault}{\color[rgb]{0,0,0}$x_3$}%
}}}}
\put(10480,119){\makebox(0,0)[lb]{\smash{{\SetFigFont{8}{9.6}{\rmdefault}{\mddefault}{\updefault}{\color[rgb]{0,0,0}$x_3$}%
}}}}
\put(12280,119){\makebox(0,0)[lb]{\smash{{\SetFigFont{8}{9.6}{\rmdefault}{\mddefault}{\updefault}{\color[rgb]{0,0,0}$x_3$}%
}}}}
\put(12730,-331){\makebox(0,0)[lb]{\smash{{\SetFigFont{8}{9.6}{\rmdefault}{\mddefault}{\updefault}{\color[rgb]{0,0,0}$x_4$}%
}}}}
\put(10030,-331){\makebox(0,0)[lb]{\smash{{\SetFigFont{8}{9.6}{\rmdefault}{\mddefault}{\updefault}{\color[rgb]{0,0,0}$x_4$}%
}}}}
\put(9600,-789){\makebox(0,0)[lb]{\smash{{\SetFigFont{8}{9.6}{\rmdefault}{\mddefault}{\updefault}{\color[rgb]{0,0,0}$\bot$}%
}}}}
\put(12300,-789){\makebox(0,0)[lb]{\smash{{\SetFigFont{8}{9.6}{\rmdefault}{\mddefault}{\updefault}{\color[rgb]{0,0,0}$\bot$}%
}}}}
\put(13200,-789){\makebox(0,0)[lb]{\smash{{\SetFigFont{8}{9.6}{\rmdefault}{\mddefault}{\updefault}{\color[rgb]{0,0,0}$\bot$}%
}}}}
\put(10930,-331){\makebox(0,0)[lb]{\smash{{\SetFigFont{8}{9.6}{\rmdefault}{\mddefault}{\updefault}{\color[rgb]{0,0,0}$x_4$}%
}}}}
\put(11830,-331){\makebox(0,0)[lb]{\smash{{\SetFigFont{8}{9.6}{\rmdefault}{\mddefault}{\updefault}{\color[rgb]{0,0,0}$x_4$}%
}}}}
\put(10500,-789){\makebox(0,0)[lb]{\smash{{\SetFigFont{8}{9.6}{\rmdefault}{\mddefault}{\updefault}{\color[rgb]{0,0,0}$\bot$}%
}}}}
\put(11400,-800){\makebox(0,0)[lb]{\smash{{\SetFigFont{8}{9.6}{\rmdefault}{\mddefault}{\updefault}{\color[rgb]{0,0,0}$\top$}%
}}}}
\put(10750,-288){\makebox(0,0)[lb]{\smash{{\SetFigFont{8}{9.6}{\rmdefault}{\mddefault}{\updefault}{\color[rgb]{0,0,0}$v$}%
}}}}
\end{picture}%

\caption{An OBDD for the boolean function $E^2_4$ respecting the variable ordering $x_1<x_2<x_3<x_4$, 
drawn (in an unreduced form) using the graphical conventions for decision diagrams \cite{Wegener00}.  Each decision node generates $6$ wires in the circuit; 
for instance, the decision node $v$ generates a $6$-wire subcircuit isomorphic to $(\neg x_4 \wedge \bot) \vee (x_4 \wedge \top)$.}
\label{fig:ein-obdd}
\end{figure}

It follows that every $P \in \mathcal{P}_n$ has a small OBDD 
with respect to every variable ordering.  

\begin{proposition}\label{prop:obdd-pi}
Let $\sigma$ be an ordering of $x_1,\ldots,x_n$ and let $P \in \mathcal{P}_n$, where $\mathcal{P}_n$ is as in (\ref{eq:family-P}).  
The $\mathcal{OBDD}_\sigma$ size of $P$ is $O(n^2)$. 
\end{proposition}
\begin{proof}
For $P_0$ and $P_n$ the statement follows directly from Proposition~\ref{prop:obdd-ein}.  
For $i=1,\ldots,n-1$ we have that $P_{i,0} \equiv E^i_n \wedge \neg x_i$ 
and $P_{i,1} \equiv E^i_n \wedge x_i$.  

Recall that if $f$ and $f'$ are boolean functions on $X$, 
and $\rho$ is any ordering of $X$, then \cite[Theorem~3.3.6]{Wegener00}:
\begin{equation}\label{eq:conj-quadr}
\mathcal{OBDD}_{\rho}(f \wedge f') \leq \mathcal{OBDD}_{\rho}(f) \cdot \mathcal{OBDD}_{\rho}(f')\text{.}
\end{equation}

Regarding the literals $\neg x_i$ 
and $x_i$ as boolean functions on $\{x_1,\ldots,x_n\}$ whose $\mathcal{OBDD}_\sigma$ size 
is constant ($6$ wires), 
the statement follows from (\ref{eq:conj-quadr}) and Proposition~\ref{prop:obdd-ein}.
\end{proof}

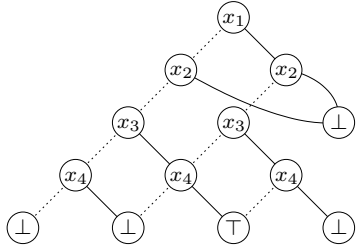
\begin{figure}[t!]
\centering
\begin{picture}(0,0)%
\includegraphics{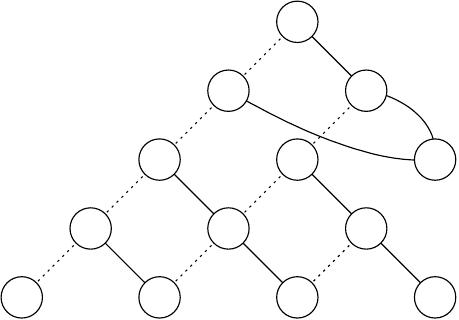}%
\end{picture}%
\setlength{\unitlength}{2901sp}%
\begingroup\makeatletter\ifx\SetFigFont\undefined%
\gdef\SetFigFont#1#2#3#4#5{%
  \reset@font\fontsize{#1}{#2pt}%
  \fontfamily{#3}\fontseries{#4}\fontshape{#5}%
  \selectfont}%
\fi\endgroup%
\begin{picture}(2986,2086)(9533,-879)
\put(11380,1019){\makebox(0,0)[lb]{\smash{{\SetFigFont{8}{9.6}{\rmdefault}{\mddefault}{\updefault}{\color[rgb]{0,0,0}$x_1$}%
}}}}
\put(10930,569){\makebox(0,0)[lb]{\smash{{\SetFigFont{8}{9.6}{\rmdefault}{\mddefault}{\updefault}{\color[rgb]{0,0,0}$x_2$}%
}}}}
\put(11830,569){\makebox(0,0)[lb]{\smash{{\SetFigFont{8}{9.6}{\rmdefault}{\mddefault}{\updefault}{\color[rgb]{0,0,0}$x_2$}%
}}}}
\put(11380,119){\makebox(0,0)[lb]{\smash{{\SetFigFont{8}{9.6}{\rmdefault}{\mddefault}{\updefault}{\color[rgb]{0,0,0}$x_3$}%
}}}}
\put(10480,119){\makebox(0,0)[lb]{\smash{{\SetFigFont{8}{9.6}{\rmdefault}{\mddefault}{\updefault}{\color[rgb]{0,0,0}$x_3$}%
}}}}
\put(10030,-331){\makebox(0,0)[lb]{\smash{{\SetFigFont{8}{9.6}{\rmdefault}{\mddefault}{\updefault}{\color[rgb]{0,0,0}$x_4$}%
}}}}
\put(9600,-789){\makebox(0,0)[lb]{\smash{{\SetFigFont{8}{9.6}{\rmdefault}{\mddefault}{\updefault}{\color[rgb]{0,0,0}$\bot$}%
}}}}
\put(12300,-789){\makebox(0,0)[lb]{\smash{{\SetFigFont{8}{9.6}{\rmdefault}{\mddefault}{\updefault}{\color[rgb]{0,0,0}$\bot$}%
}}}}
\put(12300,111){\makebox(0,0)[lb]{\smash{{\SetFigFont{8}{9.6}{\rmdefault}{\mddefault}{\updefault}{\color[rgb]{0,0,0}$\bot$}%
}}}}
\put(10930,-331){\makebox(0,0)[lb]{\smash{{\SetFigFont{8}{9.6}{\rmdefault}{\mddefault}{\updefault}{\color[rgb]{0,0,0}$x_4$}%
}}}}
\put(11830,-331){\makebox(0,0)[lb]{\smash{{\SetFigFont{8}{9.6}{\rmdefault}{\mddefault}{\updefault}{\color[rgb]{0,0,0}$x_4$}%
}}}}
\put(10500,-789){\makebox(0,0)[lb]{\smash{{\SetFigFont{8}{9.6}{\rmdefault}{\mddefault}{\updefault}{\color[rgb]{0,0,0}$\bot$}%
}}}}
\put(11400,-800){\makebox(0,0)[lb]{\smash{{\SetFigFont{8}{9.6}{\rmdefault}{\mddefault}{\updefault}{\color[rgb]{0,0,0}$\top$}%
}}}}
\end{picture}%

\caption{An OBDD for the boolean function $E^2_4 \wedge \neg x_2$ respecting the variable ordering $x_1<x_2<x_3<x_4$.}
\label{fig:ein-proj}
\end{figure}

\paragraph{SDDs vs OBDDs.}  We now prove that the hidden weighted bit function 
has small (uncompressed) SDD size; 
a slight modification of the construction, described later,   
gives the compressed case.

The key observation is that, by the definition of $\mathcal{P}_n$, 
the hidden weighted bit function $\mathrm{HWB}_n$ 
is equivalent to 
\begin{equation}\label{eq:sdd-hwb-form}
(P_0 \wedge \bot) \vee (P_n \wedge \top) \vee \bigvee_{i=1}^{n-1} ((P_{i,0} \wedge \bot) \vee (P_{i,1} \wedge \top))\text{} 
\end{equation}
because the latter is equivalent to 
$$(E^1_n \wedge x_1) \vee \cdots \vee (E^n_n \wedge x_n)$$
which is in turn equivalent to $\mathrm{HWB}_n$.  
Using the form (\ref{eq:sdd-hwb-form}), it is easy to build an SDD  
computing $\mathrm{HWB}_n$ and respecting a vtree for $\{x_1,\ldots,x_n,y\}$ 
like the one on the right in Figure~\ref{fig:struct};  
upon implementing the $P_i$'s and $P_{i,j}$'s by OBDDs, 
the construction has polynomial size by Proposition~\ref{prop:obdd-pi}.  Note that the SDD is not compressed 
because $\bot$ and $\top$ are reused $n$ times. The details follow.

\begin{theorem}\label{th:sdd-size}
The $\mathcal{SDD}$ size of $\mathrm{HWB}_n$ is $O(n^3)$.
\end{theorem}
\begin{proof}
We first define an NNF $C$ on input variables $X=\{x_1,\ldots,x_n\}$ computing (\ref{eq:sdd-hwb-form}) 
as follows.  The output gate of $C$ is a fanin $2n$ $\vee$-gate, 
with wires from $2n$ fanin $2$ $\wedge$-gates $g_0$, $g_n$, and $g_{i,j}$ for $i=1,\ldots,n-1$ and $j=0,1$.  

Let $p_0$ and $s_0$ be the two gates wiring $g_0$, 
let $p_n$ and $s_n$ be the two gates wiring $g_n$, 
and for $i=1,\ldots,n-1$ and $j=0,1$ let $p_{i,j}$ and $s_{i,j}$ be the two gates wiring $g_{i,j}$.  

Let $\sigma$ be any ordering of $x_1,\ldots,x_n$.  All the subcircuits of $C$ 
rooted at $p_0$, $s_0$, $p_n$, $s_n$, $p_{i,j}$, and $s_{i,j}$ 
($i=1,\ldots,n-1$, $j=0,1$) are OBDDs respecting the ordering $\sigma$.  Moreover:
\begin{itemize}
\item $C_{p_i}$ computes $P_i$ for $i \in \{1,n\}$;
\item $C_{p_{i,j}}$ computes $P_{i,j}$ for $i=1,\ldots,n-1$, $j=0,1$;
\item $C_{s_0}$ and $C_{s_{i,0}}$ compute $\bot$ for $i=1,\ldots,n-1$;
\item $C_{s_n}$ and $C_{s_{i,1}}$ compute $\top$ for $i=1,\ldots,n-1$. 
\end{itemize}

We prove that $C$ is an SDD respecting a suitable vtree $T$ 
for the variable set $X \cup \{y\}$.  Roughly, $T$ is 
a right-linear vtree with the exception of the variable $y$; 
see the diagram on the right in Figure~\ref{fig:struct} for the case $n=4$ 
and $\sigma=x_1<x_2<x_3<x_4$.  Formally, $T$ is defined as follows. 
Let $v$ be the root of $T$.  
The left subtree $T_l=T_{v_l}$ of $T$ 
is a right-linear vtree for $\{x_1,\ldots,x_n\}$ such that 
the variable ordering induced by its left first traversal 
is $\sigma$.  Similarly, 
the right subtree $T_r=T_{v_r}$ of $T$ 
is a vtree for $\{y\}$.  

We check that $C$ is an SDD respecting $T$.
\begin{itemize}
\item The subcircuits $C_{p_0}$, $C_{p_n}$, and $C_{p_{i,j}}$ are OBDDs respecting $\sigma$, 
and hence SDDs respecting $T_l$ ($i=1,\ldots,n-1$, $j=0,1$).  This settles (S1).
\item The subcircuits $C_{s_0}$, $C_{s_n}$, and $C_{s_{i,j}}$ are 
input gates labelled by a constant, 
and hence SDDs respecting $T_r$ ($i=1,\ldots,n-1$, $j=0,1$).  This settles (S2).
\end{itemize}
Note how the construction crucially exploits the special position of $y$ in the vtree $T$, 
while the circuit $C$ does not even read $y$. 

The partitioning properties (S3)-(S5) follow by construction and Fact~\ref{fact:part}. 
Therefore, $C$ is an SDD respecting $T$.  It remains to check that 
$C$ has size cubic in $n$.  

By construction, $C$ contains the $2n$ subcircuits $C_{p_0}$, 
$C_{p_n}$, and $C_{p_{i,j}}$ for $i=1,\ldots,n-1$ and $j=0,1$; 
each has size $O(n^2)$ by Proposition~\ref{prop:obdd-pi} 
hence, altogether, they contribute $O(n^3)$ wires in $C$.  
There remain $O(n)$ wires entering the output gate 
and the gates $g_0,g_1,\ldots,g_m$.  
\end{proof}

Combining Theorem~\ref{th:bryant} and Theorem~\ref{th:sdd-size}, 
we conclude that OBDDs and SDDs are exponentially separated 
by the hidden weighted bit function.

\paragraph{Compressed SDDs vs OBDDs.} A slight variant 
of the previous construction gives an exponential separation 
of OBDDs and compressed SDDs.  

Let $y_0,y_1,\ldots,y_n$ be fresh variables.  The boolean function 
$F_n$ of the variables $x_1,\ldots,x_n,y_0,y_1,\ldots,y_n$, 
called \emph{generalized hidden weighted bit} function, is defined by
\begin{equation}\label{eq:Fn}
(P_0 \wedge \neg y_0) \vee (P_n \wedge y_n) \vee \bigvee_{i=1}^{n-1} ((P_{i,0} \wedge \neg y_i) \vee (P_{i,1} \wedge y_i))\text{.} 
\end{equation}

Notice that the form (\ref{eq:Fn}) is exactly as the form (\ref{eq:sdd-hwb-form}), 
except that the $n$ copies of $\bot$ and the $n$ copies of $\top$ 
are replaced by the $2n$ pairwise nonequivalent formulas $\neg y_0$, 
$y_n$, $y_i$, and $\neg y_i$ ($i=1,\ldots,n-1$), 
so that (\ref{eq:Fn}) has indeed a compressed SDD implementation.   The details follow.

\begin{lemma}\label{lemma:sdd-size}
The $\mathcal{SDD}_c$ size of $F_n$ is $O(n^3)$.
\end{lemma}
\begin{proof}
We construct an NNF $C$ on input variables $X=\{x_1,\ldots,x_n,y_0,y_1,\ldots,y_n\}$ 
computing (\ref{eq:Fn}) along the lines of Theorem~\ref{th:sdd-size}.  The only modification 
is that $C_{s_0}$ is an input gate labelled $\neg y_0$, 
$C_{s_n}$ is an input gate labelled $y_n$, 
$C_{s_{i,0}}$ is an input gate labelled $\neg y_i$, 
and $C_{s_{i,1}}$ is an input gate labelled $y_i$ ($i=1,\ldots,n-1$).

We claim that $C$ is a compressed SDD respecting 
a vtree $T$ for the variable set $X$ built exactly as in Theorem~\ref{th:sdd-size} 
except that the right subtree $T_r=T_{v_r}$ of $T$ 
is a right-linear vtree for $\{y_0,y_1,\ldots,y_n\}$ such that 
the variable ordering induced by its left first traversal is $\rho$.  
See Figure~\ref{fig:vt-sdd} for the case $n=4$, 
$\sigma=x_1<\cdots<x_4$, and $\rho=y_0<y_1<\cdots<y_4$.  

To check that $C$ is a compressed SDD respecting $T$, 
notice that the subcircuits $C_{p_0}$ and $C_{p_{i,j}}$ are OBDDs respecting $\sigma$, 
and hence compressed SDDs respecting $T_l$ ($i=1,\ldots,n$, $j=0,1$), 
and the subcircuits $C_{s_0}$ and $C_{s_{i,j}}$ are OBDDs respecting $\rho$, 
and hence compressed SDDs respecting $T_r$ ($i=1,\ldots,n$, $j=0,1$).  
Moreover, it is easily verified that the output gate of $C$ is compressed as by condition (C).  
Hence $C$ is compressed.  The rest of the proof is identical to that of Theorem~\ref{th:sdd-size}.
\end{proof}

\begin{figure}[t!]
\centering
\begin{picture}(0,0)%
\includegraphics{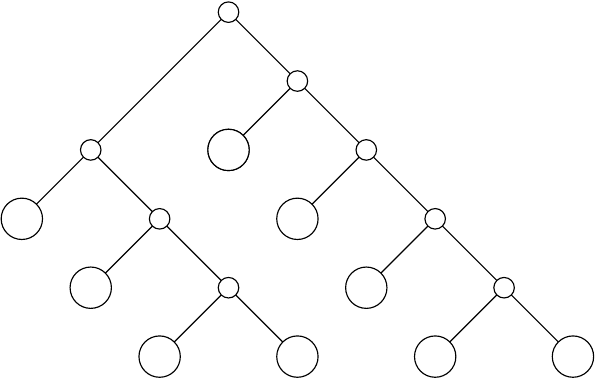}%
\end{picture}%
\setlength{\unitlength}{2901sp}%
\begingroup\makeatletter\ifx\SetFigFont\undefined%
\gdef\SetFigFont#1#2#3#4#5{%
  \reset@font\fontsize{#1}{#2pt}%
  \fontfamily{#3}\fontseries{#4}\fontshape{#5}%
  \selectfont}%
\fi\endgroup%
\begin{picture}(3886,2468)(10883,-429)
\put(11830,-331){\makebox(0,0)[lb]{\smash{{\SetFigFont{8}{9.6}{\rmdefault}{\mddefault}{\updefault}{\color[rgb]{0,0,0}$x_3$}%
}}}}
\put(10930,569){\makebox(0,0)[lb]{\smash{{\SetFigFont{8}{9.6}{\rmdefault}{\mddefault}{\updefault}{\color[rgb]{0,0,0}$x_1$}%
}}}}
\put(11380,119){\makebox(0,0)[lb]{\smash{{\SetFigFont{8}{9.6}{\rmdefault}{\mddefault}{\updefault}{\color[rgb]{0,0,0}$x_2$}%
}}}}
\put(12730,-331){\makebox(0,0)[lb]{\smash{{\SetFigFont{8}{9.6}{\rmdefault}{\mddefault}{\updefault}{\color[rgb]{0,0,0}$x_4$}%
}}}}
\put(12747,580){\makebox(0,0)[lb]{\smash{{\SetFigFont{8}{9.6}{\rmdefault}{\mddefault}{\updefault}{\color[rgb]{0,0,0}$y_1$}%
}}}}
\put(13197,130){\makebox(0,0)[lb]{\smash{{\SetFigFont{8}{9.6}{\rmdefault}{\mddefault}{\updefault}{\color[rgb]{0,0,0}$y_2$}%
}}}}
\put(13647,-320){\makebox(0,0)[lb]{\smash{{\SetFigFont{8}{9.6}{\rmdefault}{\mddefault}{\updefault}{\color[rgb]{0,0,0}$y_3$}%
}}}}
\put(14547,-320){\makebox(0,0)[lb]{\smash{{\SetFigFont{8}{9.6}{\rmdefault}{\mddefault}{\updefault}{\color[rgb]{0,0,0}$y_4$}%
}}}}
\put(12297,1030){\makebox(0,0)[lb]{\smash{{\SetFigFont{8}{9.6}{\rmdefault}{\mddefault}{\updefault}{\color[rgb]{0,0,0}$y_0$}%
}}}}
\end{picture}%

\caption{The vtree for $F_4$ in the proof of Lemma~\ref{lemma:sdd-size}.}
\label{fig:vt-sdd}
\end{figure}

We now prove that the generalized hidden weighted bit function $F_n$ needs large OBDDs.

\begin{lemma}\label{lemma:obdd-size}
The $\mathcal{OBDD}$ size of $F_n$ is $2^{\Omega(n)}$.
\end{lemma}
\begin{proof}
Let $N$ be the size of a smallest $\mathcal{OBDD}$ on variables 
$X=\{x_1,\ldots,x_n,y_0,y_1,\ldots,y_n\}$ computing $F_n$, 
and let $\rho$ be any ordering of $X$ 
such that $\mathcal{OBDD}_\rho(F_n)=N$.  

Let $G_n(x_1,\ldots,x_n)$ be the subfunction of $F_n$ where 
$y_0,y_1,\ldots,y_n$ are replaced by $1$, in symbols:
\begin{equation}\label{eq:Gn}
G_n \equiv {F_n}(x_1,\ldots,x_n,1,1,\ldots,1)\text{.} 
\end{equation}

Since conditioning (unboundedly many variables of) an OBDD 
does not increase its size \cite[Theorem~2.4.1]{Wegener00}, we have that 
\begin{equation}\label{eq:step-1}
\mathcal{OBDD}_{\rho}(G_n) \leq \mathcal{OBDD}_{\rho}(F_n)\text{.}
\end{equation}

We now claim that $G_n$ is the hidden weighted bit function on $n$ variables. 
Indeed, by construction, 
\begin{align*}
G_n &\equiv  F_n(x_1,\ldots,x_n,1,1,\ldots,1) \\
 & \equiv P_n \vee \bigvee_{i=1}^{n-1} P_{i,1}
\end{align*}
which we already observed being equivalent to $\mathrm{HWB}_n$.  Therefore $\mathcal{OBDD}(G_n)=2^{\Omega(n)}$ by Theorem~\ref{th:bryant}, 
and in particular $\mathcal{OBDD}_\rho(G_n) \geq 2^{\Omega(n)}$.  By (\ref{eq:step-1}), 
we are done.
\end{proof}

An exponential separation of OBDDs and compressed SDDs follows.

\begin{theorem}\label{th:exp-sep}
There exists an unbounded arity class of boolean functions $\mathcal{F}$ 
such that every arity $n$ function $f \in \mathcal{F}$ has $\mathcal{SDD}_c$ size in $O(n^3)$ 
and $\mathcal{OBDD}$ size in $2^{\Omega(n)}$.
\end{theorem}
\begin{proof}
Take $\mathcal{F}=\{ F_m \colon m \in \mathbb{N} \}$, 
where $F_m$ is as in (\ref{eq:Fn}).  Then $F_m$ has 
compressed $\mathcal{SDD}$ size $O(m^3)$ by Lemma~\ref{lemma:sdd-size} 
and $\mathcal{OBDD}$ size $2^{\Omega(m)}$ by Lemma~\ref{lemma:obdd-size}.  
Since $F_m$ has $n=2m+1$ 
variables, it follows that $F_m$ has $\mathcal{SDD}_c$ size in $O(n^3)$ 
and $\mathcal{OBDD}$ size in $2^{\Omega(n)}$.
\end{proof}

Notably, the function class giving the exponential separation 
is as hard on compressed SDDs as the function class giving the quasipolynomial separation 
(cubic in both cases, see Theorem~\ref{th:raz}).

\section{Discussion}\label{sect:disc}

We have shown that OBDDs and SDDs 
are exponentially separated by 
the hidden weighted bit function, 
while OBDDs and compressed SDDs 
are exponentially separated by the generalized hidden weighted bit function, 
$F_n$ in (\ref{eq:Fn}), that contains the hidden weighted bit function as a subfunction:
\begin{equation}\label{eq:Fn-proj}
F_n(x_1,\ldots,x_n,1,1,\ldots,1)=\mathrm{HWB}_n(x_1,\ldots,x_n)\text{.}
\end{equation}
Separating OBDDs and SDDs by the hid\-den weight\-ed bit function, instead of by a function designed adhoc,  
further corroborates the theoretical quality of SDDs.  As articulated by Bollig et al.\ \cite{BolligLSW99}, 
any useful extension of OBDDs is expected to implement the hidden weighted bit function efficiently.  

The SDD $C$ described in the proof of Theorem~\ref{th:sdd-size} is not compressed, 
because $\bot$ and $\top$ are reused $n$ times.  
In view of the canonical construction of an SDD over a vtree \cite[Theorem~3]{Darwiche11}, it is readily observed that 
compressing $C$ with respect to the vtree $T$ in the proof of Theorem~\ref{th:sdd-size} implies finding 
a small SDD for $\mathrm{HWB}_n$ with respect to the left subtree of $T$, 
that is, a small OBDD for $\mathrm{HWB}_n$; but this is impossible by Theorem~\ref{th:bryant}.  The fact 
that compressing an SDD over its vtree 
may increase the size exponentially has been observed already \cite[Theorem~1]{BroeckDarwiche15}.  
We reiterate the observation here only because our argument is significantly shorter.

We conclude mentioning a nonobvious, and perhaps even unexpected, 
aspect of our separation result.  An inspection of our construction 
shows that SDDs are already exponentially more succinct than \emph{general OBDDs} 
even allowing only \emph{one sentential decision} (and possibly many Shannon decisions); recall (\ref{eq:sdd-hwb-form}) and (\ref{eq:Fn}).  
The construction by Xue et al.\ \cite{XueCD12} already uses \emph{nested sentential decisions} 
even to separate \emph{OBDDs over a fixed variable ordering} from SDDs!

\paragraph{Questions.} We do not know whether the hidden weighted bit function has superpolynomial compressed 
SDD size for all vtrees; a positive answer would separate compressed and uncompressed SDDs  
in succinctness and, in view of Lemma~\ref{lemma:sdd-size} and (\ref{eq:Fn-proj}), would prove that 
compressed SDDs do not support conditioning (of unboundedly many variables) in polynomial size.

In view of Theorem~\ref{th:raz}, it is natural to ask which SDDs are quasipolynomially simulated by OBDDs.  
Our separating family shows that SDDs with unbounded fanin disjunctions cannot be quasipolynomially simulated by OBDDs.  
On the other hand, recent work by Darwiche and Oztok essentially shows that 
SDDs over binary disjunctions (fanin $2$) admit a quasipolynomial simulation by OBDDs \cite[Theorem~1]{OztokDarwiche15}.  
In this light, it is tempting to conjecture that the above criterion is exact, that is, every SDD class over bounded fanin disjunctions 
does indeed admit a quasipolynomial simulation by OBDDs.

Finally, a natural question arising in the context of the present work 
is about the relative succinctness of SDDs and structured deterministic NNFs (see (\ref{eq:inclusions})); 
to the best of our knowledge, the question is open. By Theorem~\ref{th:sdd-size}, 
at least we now know that the hidden weighted bit function is not a candidate to separate the two classes.

\section*{Acknowledgments}

The author thanks Igor Razgon for generously introducing 
him to the problem addressed in this article, and an 
anonymous reviewer for suggesting the comparison 
with \cite{XueCD12} discussed in the conclusion.  This research was supported by the FWF Austrian Science Fund 
(Parameterized Compilation, P26200).

\DeclareRobustCommand{\DE}[3]{#3}

\end{document}